\documentclass[11pt]{article}
\pdfoutput=1

\usepackage[affil-it]{authblk}

\usepackage{silence}

\usepackage{amsmath,amsthm,amssymb,thmtools,thm-restate,dsfont}

\usepackage[english]{babel}

\usepackage{bbm}

\usepackage{constants}

\usepackage{quoting}

\usepackage[margin=1in]{geometry}

\newif\ifdraft
\drafttrue

\usepackage{todonotes}

\usepackage[inline,shortlabels]{enumitem}

\usepackage[fullversion]{auxkling}

\xpretocmd{\footnote}{\unskip}{}{}

\usepackage[normalem]{ulem}

\theoremstyle{plain}
\newtheorem{theorem}{Theorem}
\newtheorem*{theorem-non}{Theorem}

\newtheorem{proposition}{Proposition}
\newtheorem{lemma}{Lemma}

\theoremstyle{definition}
\newtheorem{definition}{Definition}

\theoremstyle{remark}

\newtheorem{conjecture}{Conjecture}

\newcommand*{\hmajority}[1]{\texorpdfstring{$#1$}{#1}-\textsc{Majority}\xspace}

\newcommand*{\twochoice}{2-\textsc{Choices}\xspace}
\newcommand*{\voter}{\textsc{Voter}\xspace}
\newcommand*{\voterAlternative}{\textsc{Polling}\xspace}

\def\AC/{\textsc{AC}}
\def\AClong/{Anonymous Consensus}
\def\aclong/{anonymous consensus}

\newcommand*{\aggfvoter}{\alpha^{({\mathcal{V}})}}
\newcommand*{\aggfhmajority}[1]{\alpha^{(#1M)}}

\newcommand*{\sortDown}[1]{{#1}^{\downarrow}}

\newcommand*{\Mult}{\operatorname{Mult}}

\newcommand\E[1]{\mathbb{E}\left[\,#1\,\right]}

\title{Ignore or Comply? On Breaking Symmetry in Consensus}

\author[3,5]{Petra Berenbrink}
\author[4]{Andrea Clementi}
\author[6]{Robert Elsässer}
\author[5]{Peter Kling\thanks{%
    E-mail: \email{peter.kling@uni-hamburg.de};
    Phone: +49 40 42883-2407
    }
}
\author[1,3]{Frederik Mallmann-Trenn}
\author[2]{Emanuele Natale}

\affil[1]{École normale supérieure, Paris, France}
\affil[2]{Max Planck Institute for Informatics, Saarbrücken, Germany}
\affil[3]{Simon Fraser University, Burnaby, Canada}
\affil[4]{Università di Roma Tor Vergata, Rome, Italy}
\affil[5]{University of Hamburg, Hamburg, Germany}
\affil[6]{University of Salzburg, Salzburg, Austria}

\begin{document}
\maketitle

\thispagestyle{empty}

\begin{abstract}

We study \emph{consensus processes} on the complete graph of $n$ nodes.
Initially, each node supports one from opinion from a set of up to $n$ different opinions.
 Nodes randomly and in
parallel sample the opinions of constant many nodes. Based on these samples,
they use an \emph{update rule} to change their own opinion. The goal is to reach
\emph{consensus}, a configuration where all nodes support the same opinion.

We compare two well-known update rules: \twochoice and \hmajority{3}. In the
former, each node samples two nodes and adopts their opinion if they agree. In
the latter, each node samples three nodes: If an opinion is supported by at
least two samples the node adopts it, otherwise it randomly adopts one of the
sampled opinions.
Known results for these update rules focus on initial configurations with a
limited number of colors (say $n^{\frac 13}$), or typically assume a bias,
where one opinion has a much larger support than any other. For such biased 
configurations, the time to reach consensus is roughly the same for \twochoice
and \hmajority{3}. 

Interestingly, we prove that this is no longer true for configurations with a
large number of initial colors. In particular, we show that \hmajority{3}
reaches consensus with high probability in $\LDAUOmicron[small]{n^{3/4} \cdot
\log^{7/8} n}$ rounds, while \twochoice can need $\LDAUOmega{n / \log n}$
rounds. We thus get  the first unconditional sublinear bound for \hmajority{3}
and the first result separating the consensus time of these processes. Along
the way, we develop a framework that allows a fine-grained comparison between
consensus processes from a specific class. We believe that this framework might
help to classify the performance of more consensus processes.

\bigskip


\end{abstract}

\clearpage
\setcounter{page}{1}



\textcolor{red}{
}

\section{Introduction}

We study \emph{consensus} (also known as agreement) processes resulting from
executing a simple algorithm in a distributed system. The system consists of $n$
anonymous nodes connected by a complete graph. Initially, each node supports one
opinion from the set $\intcc{k} \coloneqq \set{1, \dots, k}$. We refer to these
opinions as \emph{colors}. The system state is modeled as a \emph{configuration}
vector $c \in \N_0^k$, whose $i$-th component $c_i$ denotes the number (support) of nodes
with color $i$. A consensus process is specified by an \emph{update rule} that
is executed by each node. The goal is to reach a configuration in which all
nodes support the same color; the special case where nodes have pairwise
distinct colors is leader election, an important primitive in distributed
computing. We assume a severely restricted synchronous communication mechanism
known as \emph{Uniform Pull}~\cite{DGHILSSST87, KSSV00, KDG03}. Here, in each
discrete round, nodes independently pull information from some (typically
constant) number of randomly sampled nodes. Both the message sizes and the
nodes' local memory must be of size $\LDAUOmicron{\log k}$.

The so-called \voter process (also known as \voterAlternative), uses the most
naïve update rule: In every round, each node samples one neighbor independently
and uniformly at random and adopts that node's color. Two further natural and
prominent consensus processes are the \emph{\twochoice} and the
\emph{\hmajority{3}} process. Their corresponding update rules, as executed
synchronously by each node, are as follows:
\begin{itemize}
\item \twochoice:
    Sample two nodes independently and uniformly at random. If the samples have
    the same color, adopt it. Otherwise, ignore them and keep your current
    color.

\item \hmajority{3}:
    Sample three nodes independently and uniformly at random. If a color is
    supported by at least two samples, adopt it. Otherwise, adopt the color of
    one of them at random\footnote{%
        Equivalently, the node may adopt the color of a fixed sample (the first, or second, or third).
    }.

\end{itemize}
One reason for the interest in these processes is that they represent simple and
efficient self-stabilizing solutions for \emph{Byzantine agreement}~\cite{PSL80,
R83}: achieving consensus in the presence of an adversary that can disrupt a
bounded set of nodes each round~\cite{BCNPST14, BCNPT16, CER14, EFKMT16}.
Further interest stems from the fact that they capture aspects of how agreement
is reached in social networks and biological systems~\cite{Dolev, CER14,
HouseHunt}.

At first glance, the above processes look quite different. But a slight
reformulation of \hmajority{3}'s update rule reveals an interesting connection:
\begin{itemize}
\item \hmajority{3} (alt.):
    Sample two nodes independently and uniformly at random. If the samples
    have the same color, adopt it. Otherwise, sample a new neighbor and adopt
    its color.
\end{itemize}
This highlights the fact that \hmajority{3} is a combination of \twochoice and
\voter: Each node $u$ performs the update rule of \twochoice. If the sampled
colors do not match, instead of keeping its color, $u$ executes the update rule
of \voter. Interestingly enough, both \hmajority{3} and \twochoice behave
identical in expectation\footnote{%
    Simple calculations~\cite{BCNPST14, EFKMT16} show that, for both processes, if $x_i$ is the
    current fraction of nodes with color $i$ then    the expected fraction of nodes
    with color $i$ after one round is $x_i^2 + (1 - \sum x_j^2) \cdot x_i$.
}. In comparison to \voter, both \twochoice and \hmajority{3} exhibit a drift:
they favor colors with a large support, for which it is more likely that the
first two samples match. In particular, if there is a certain initial
bias\footnote{The \emph{bias} is the difference between the number of nodes
supporting the most and second most common color.} towards one color, \voter
still needs linear time (in $n$) to reach consensus, while both \twochoice and
\hmajority{3} can exploit the bias to achieve sublinear time. On the other hand,
it is unknown how \twochoice and \hmajority{3} behave when they start from
configurations having a large number of colors and no (or small) bias (see
\cref{sec:related_work} for details).

In this paper, we compare the time \twochoice and \hmajority{3} require to reach
consensus. In particular, we prove that there is a polynomial gap between their
performance if the bias is small 
and the number of
colors is large (say $\LDAUOmega{n^{1/3}}$). This result follows from our unconditional sublinear bound for \hmajority{3} and an almost linear lower bound
for \twochoice. Details of our results and contribution are given in
\cref{sec:contribution}.


\subsection{Related Work}
\label{sec:related_work}

Consensus processes are a quite general model that can be used to study and
understand different flavors of \enquote{spreading phenomena}, for example the
spread of infectious diseases, rumors and opinions in societies, or
viruses/worms in computer networks. Apart from the above mentioned processes,
such spreading processes include the \textsc{Moran} process~\cite{LHN05, L14},
contact processes, and classic epidemic processes~\cite{BG90, L99, nature16}.
This overview concentrates on results concerned with the time needed by \voter,
\twochoice, and \hmajority{3} to reach consensus. We also provide a short
comparison of these processes and briefly discuss some of the slightly more
distant relatives at the end of this section.

\paragraph{\voter}
Previous work provides strong results for the time to consensus of the \voter
process, even for arbitrary graphs. These results exploit an interesting
duality: the time reversal of the \voter process turns out to be the coalescing
random walk process (see~\cite{HP01, AF14}). The expected runtime of \voter on
the complete graph and with nodes of pairwise distinct colors is
$\LDAUTheta{n}$. This follows easily from results for more general graphs and
the above mentioned duality: The authors of~\cite{CEOR13} provide the upper
bound $\LDAUOmicron{\mu^{-1} \cdot (\log^4n+ \rho)}$ on the expected coalescing
time. Here $\mu$ is the graph's spectral gap and $\rho = (d_{\operatorname{avg}}
\cdot n)^2 / \sum_{u\in V} d^2(u)$, where $d_{\operatorname{avg}}$ denotes the
average degree. In~\cite{BGKM16} the authors show that the expected time to
consensus is bounded by $\LDAUOmicron{m / (d_{\min}\cdot\phi)}$. Here, $m$ is
the number of edges, $\phi$ the conductance, and $d_{\min}$ the minimal degree.

\paragraph{\twochoice}
To the best of our knowledge, the only work that considers the case $k > 2$ for
\twochoice is~\cite{EFKMT16}. The authors study the complete graph and show that
\twochoice reaches consensus with high probability in $\LDAUOmicron{k \cdot \log
n}$ rounds, provided that $k = \LDAUOmicron{n^\varepsilon}$ for a small constant
$\varepsilon > 0$ and an initial bias of $\LDAUOmega{\sqrt{n \cdot \log n}}$.
In~\cite{CER14, CERRS15}, the authors consider \twochoice for $k = 2$ on
different graphs. For random $d$-regular graphs,~\cite{CER14} proves that all
nodes agree on the first color in $\LDAUOmicron{\log n}$ rounds, provided that
the bias is $\LDAUOmega[small]{ n \cdot \sqrt{1/d + d/n} }$. The same holds for
arbitrary $d$-regular graphs if the bias is $\LDAUOmega{\lambda_2 \cdot n}$,
where $\lambda_2$ is the second largest eigenvalue of the transition matrix. In
\cite{CERRS15}, these results are extended to general expander graphs.

\paragraph{\hmajority{3}}
All theoretical results for \hmajority{3} consider the complete graph. The
authors of~\cite{BCNPST14} assume that the bias is $\LDAUOmega[small]{ \min\set{
\sqrt{2k}, {(n / \log n)}^{1/6} } \cdot \sqrt{ n \log n } }$. Under this
assumption, they prove that consensus is reached with high probability in
$\LDAUOmicron[small]{\min\set{ k, {(n / \log n)}^{1/3} } \cdot\log n}$ rounds,
and that this is tight if $k \leq {(n / \log n)}^{1/4}$. The only result without
bias~\cite{BCNPT16} restricts the number of initial colors to $k =
\LDAUomicron[small]{n^{1/3}}$. Under this assumption, they prove that
\hmajority{3} reaches consensus with high probability in
$\LDAUOmicron[small]{(k^2 {(\log n)}^{1/2}+ k \log n) \cdot (k + \log n)}$
rounds. Their analysis considers phases of length $\LDAUOmicron[small]{k^2\log n}$ and shows that,
at the end of each phase,  one of the initial colors disappears with high
probability. Note that this approach – so far the only one not assuming any bias
– cannot yield sublinear bounds with respect to $k$.

\bigbreak
Comparing the above processes on the complete graph for $k > 2$, we see that
there are situations where \voter is much slower than \twochoice or
\hmajority{3}. Even with a linear bias, \voter is known to have linear runtime.
In contrast, whenever there is a sufficient bias towards one color, both \twochoice and
\hmajority{3} can exploit this to achieve sublinear runtime\footnote{%
    Interestingly enough, \twochoice converges to the relative majority color whenever
    the initial bias is
    $\LDAUOmega[small]{\sqrt{n \log n}}$~\cite{BGKM16}, while the lowest bias
    that \hmajority{3} can cope with in order to converge to majority is 
    $\LDAUOmega[small]{\sqrt{k} \cdot \sqrt{n \log n}}$~\cite{BCNPST14}.
}. However, to the best of our knowledge there are no unconditional results on
\twochoice and \hmajority{3}. All but one results need a minimum bias (at least
$\LDAUOmega{\sqrt{n \log n}}$), and the only approach that works without any
bias restricts the number of colors to $k = \LDAUomicron{ n^{1/3}
}$~\cite{BCNPT16}.

A related consensus process is
\emph{$2$-\textsc{Median}}~\cite{DBLP:conf/spaa/DoerrGMSS11}. Here, every node
updates its color (a numerical value) to the median of its own value and two
randomly sampled nodes. Without assuming any initial bias, the authors show that
this process reaches consensus with high probability in $\LDAUOmicron{\log k
\cdot \log\log n + \log n}$ rounds. This is seemingly stronger than the bounds
achieved for \hmajority{3} and \twochoice without bias. However, it comes at the
price of a complete order on the colors (our processes require colors only to be
testable for identity). Moreover, $2$-\textsc{Median} is not self-stabilizing
for Byzantine agreement (unlike \hmajority{3} and \twochoice~\cite{BCNPT16,
EFKMT16}): it cannot guarantee \emph{validity}\footnote{%
    Byzantine agreement requires that the system does not converge to a color
    that was initially not supported by at least one non-corrupted node.
}~\cite{BCNPT16}. Another consensus process is the
\emph{\textsc{UndecidedDynamics}}. Here, each node randomly samples one neighbor
and, if the sample has a different color, adopts a special \enquote{undecided}
color. In subsequent rounds, it tries to find a new (real) color by sampling one
random neighbor. The most recent results~\cite{DBLP:conf/soda/BecchettiCNPS15}
show that, for a large enough bias, consensus is reached with high probability
in at most $\LDAUOmicron{k \log n}$ rounds. Slightly more involved variants yield
improved bounds of $\LDAUOmicron{\log k \cdot \log
n}$~\cite{Petra-Berenbrink:2016aa, GP16, EFKMT16}. 
However, observe that for $k = {n}$ all nodes become undecided with constant
probability instead of agreeing on a color.


\subsection{Contribution \& Approach}
\label{sec:contribution}

In this work, we provide  an upper bound on \hmajority{3}
and a lower bound on \twochoice that solve two open issues: We give the first
unconditional sublinear bound on any of these processes (an open issue from,
e.g., \cite{BCNPT16}) and prove that there can be a polynomial gap between the
performance of \hmajority{3} and \twochoice (see \cref{thm:main:simpl}
below). One should note that this gap is in stark contrast not only to the
expected behavior of both processes (which is identical)
but also to the setting when there is a bias towards one color (where both
processes exhibit the same asymptotic runtime $\LDAUOmicron{k \cdot \log n}$;
see \cref{sec:related_work}).

The following theorem states slightly simplified   versions of
our upper and lower bounds (see \cref{thm:3maj} and
\cref{lem:lowerTCstrong} for the detailed statements).
\begin{theorem}[Simplified]
\label{thm:main:simpl}
Starting from an arbitrary configuration, \hmajority{3} reaches consensus with
high probability in $\LDAUOmicron[small]{n^{3/4}\log^{7/8} n }$ rounds. When
started from a configuration where each color is supported by at most
$\LDAUOmicron{\log n}$ nodes, \twochoice needs with high probability
$\LDAUOmega{n / \log n}$ rounds to reach consensus.
\end{theorem}
The lower bound for \twochoice follows mostly by standard techniques, using a
coupling with a slightly simplified process and Chernoff bounds. The proof of
the upper bound for \hmajority{3} is more involved and based on a combination of
various techniques and results from different contexts. This approach not only
results in a concise proof of the upper bound, but yields some additional,
interesting results along the way. We give a brief overview of our approach in
the next paragraph.

\paragraph{Approach}
To derive our upper bound on the time to consensus required by \hmajority{3}, we
split the analysis in two phases:
\begin{enumerate*}[(a)]
\item the time needed to go from $n$ to $\approx n^{1/4}$ colors and
\item the time needed to go from $\approx n^{1/4}$ to one color.
\end{enumerate*}
The runtime of the second phase follows by a simple application
of~\cite{BCNPT16} and is $\SOFTOmicron[small]{n^{3/4}}$. Bounding the runtime of
the first phase is more challenging: we cannot rely on the drift from a bias or
similar effects, and it is not clear how to perform a direct analysis in this
setting (\hmajority{3} is geared towards biased configurations). To overcome
this issue, we resort to a coupling between \voter and \hmajority{3}. Since the
construction of such a coupling seems elusive, we use some machinery from
majorization theory~\cite{Marshall:2011aa} to merely prove the \emph{existence} of the coupling 
(see next paragraph). As a consequence of (the existence of) this coupling, we
get that the time needed by \hmajority{3} to reduce the number of colors to a
fixed value is stochastically dominated by the time \voter needs for this
(\cref{cor:3-majvspolling}). This, finally, allows us to upper bound the time
needed by \hmajority{3}\footnote{%
    Note that for a large number of colors, a node executing \hmajority{3}
    behaves with high probability like a node performing \voter. Thus, it is
    relatively tight to bound \hmajority{3} by \voter in this parameter regime.
} to go from $\approx n$ to $\approx n^{1/4}$ colors by the time \voter needs
for this (which, in turn, we bound in \cref{lem:reducetok} by $\SOFTOmicron{n
/ k}$).

The technically most interesting part of our analysis is the proof of the
stochastic dominance between \hmajority{3} and \voter. It works for a wide class
of processes (including \voter and \hmajority{3}), which we call \emph{\aclong/}
(\AC/-) processes (see \cref{def:anoagg_processes}). 
These are defined by an update rule that causes each node to
adopt any color $i$ with the same probability $\alpha_i$ that depends only on
the current frequency of colors. 

In the following, we provide a natural 
way to compare two processes. First, we define a way to compare two
configurations $c$ and $c'$. 
We use
\emph{vector majorization} for this purpose: $c$ majorizes $c'$ ($c \succeq c'$)
if the total support of the $j$ largest colors in $c$ is not smaller than that
in $c'$ for all $j \in \intcc{k}$. In particular, note that a configuration
where all nodes have the same color majorizes any other configuration. 
Let us write $P(c)$ for the (random) configuration obtained by performing one
step of a process $P$ on configuration $c$. Consider two processes $P, P'$
and two configurations $c, c'$ with $c \succeq c'$.
We say $P$ \emph{dominates} $P'$ if, for all $j \in \intcc{k}$, the
following holds:
\begin{quoting}[leftmargin=2.9em]
    The sum of the $j$ largest components of the vector $\E{P(c)}$ is not
    smaller than that of $\E{P'(c)}$.
\end{quoting}
Note that this definition is not restricted to \AC/-processes.

Our main technical result (\cref{cor:anonymous_coupling}) proves that, for
two \AC/-processes, $P$ dominating $P'$ implies that the time needed by $P'$ to
reduce the number of colors to a fixed value stochastically dominates the time
$P$ needs for this. Note that while this statement might sound obvious, it is
not true in general (if one of the processes is not an \AC/-process): \twochoice
dominates \voter, but it is much slower in reducing the number of colors when
there are many colors.


\section{Consensus Model \& Technical Framework}
\label{sec:model_and_framework}

This section introduces our technical framework using concepts from
majorization theory, which is used in \cref{sec:upper_bound_3majority} to
derive the sublinear upper bound on \hmajority{3}. After defining the model and
general notation, we provide a few definitions and state the main result of
this section (\cref{cor:anonymous_coupling}).

\subsection{Model and Notation}
\label{sec:model_notation}

We consider the \emph{consensus problem} on the complete graph of $n$ nodes.
Initially, each node supports one opinion (or color) from the set $\intcc{k}
\coloneqq \set{1, 2, \dots, k}$, where $k \leq n$. Nodes interact in
synchronous, discrete rounds using the \emph{Uniform Pull}
mechanism~\cite{DGHILSSST87}. That is, during every round each node can ask the
opinion of a constant number of random neighbors. Given these opinions, it
updates its own opinion according to some fixed update rule. The goal of the
system is to reach \emph{consensus} (a configuration where all nodes support the
same opinion).

Let $\N \coloneqq \set{1, 2, \dots}$ and $\N_0 \coloneqq \set{0, 1, \dots}$. We
describe the system state after any round by an $n$-dimensional integral vector
$c = \intoo{c_i}_{i \in \intcc{n}} \in \N_0^n$ with $\sum_{ i \in \intcc{n} }
c_i = n$. Here, the $i$-th component $c_i \in \N_0$ corresponds to the number of
nodes supporting opinion $i$. If $k < n$, then $c_i = 0$ for all $i \in
\set{k+1, k+2, \dots, n}$. We use $\mathcal{C} \coloneqq \set{ c\in\N_0^n |
\sum_{ i \in \intcc{k} } c_i = n }$ to denote the set of all possible
configurations.

Let $d \in \N$ and $x, y \in \R^d$. We define $\norm{x}_1 \coloneqq \sum_{i \in
\intcc{d} } x_i$ and $\norm{x}_2 \coloneqq \smash{{\bigl( \sum_{ i \in \intcc{d}
} x_i^2 \bigr)}^{1/2}}$. Moreover, let $\smash{\sortDown{x}}$  denote a
permutation of $x$ such that all components are sorted non-increasingly.
 We write $x \succeq y$ and say $x$ \emph{majorizes}
$y$ if, for all $l \in \intcc{d}$, we have $\sum_{ i \in \intcc{l} }
\sortDown{x}_i \geq \sum_{ i \in \intcc{l} } \sortDown{y}_i$ and $\norm{x}_1 =
\norm{y}_1$. 
For two random variables $X$ and $Y$ we write $X \leq^{\text{st}} Y$ if $X$ is  \emph{stochastically dominated} by $Y$, i.e., $\Pr{X \geq x} \leq \Pr{Y \geq x}$ for all $x\in \N_0$.
A function $f\colon \R^d \to \R$ is \emph{Schur-convex} if $x
\succeq y \Rightarrow f(x) \geq f(y)$. For a probability vector $\Theta \in
\intcc{0,1}^d$, we use $\Mult(m,\Theta)$ to denote the multinomial distribution
for $m$ trials and $d$ categories (the $i$-th category having probability
$\Theta_i$).

\subsection{Comparing \AClong/ Processes}
\label{sec:anoagg}

We first define a class of processes defined by update rules that depend only on
the current configuration.
The update rule states that each nodes adopts a color
$i$ with the same probability $\alpha_i(c)$, where $c \in \mathcal{C}$ is the
current configuration. In particular, node IDs (including the sampling node's
ID) do not influence the outcome. In this sense, such update rules are
\emph{anonymous}.
\begin{definition}[\AClong/ Processes]
\label{def:anoagg_processes}
Given a distributed system of $n$ nodes, an \emph{\aclong/ process} $P_{\alpha}$
is characterized by a \emph{process function} $\alpha\colon \mathcal{C} \to
\intcc{0,1}^n$ with $\sum_{i\in\intcc{n}}\alpha_i(c) = 1$ for all $c \in
\mathcal{C}$. When in configuration $c \in \mathcal{C}$, each node independently
adopts opinion $i\in\intcc{k}$ with probability $\alpha_i(c)$. We use the
shorthand \emph{\AC/-processes} to refer to this class.
\end{definition}
Given an \AC/-process $P_{\alpha}$ and a fixed initial configuration,
let\footnote{%
    Recall that, with a slight abuse of notation we also write $P(c)$ for the
    (random) configuration obtained by performing one step of a process $P$ on
    configuration $c$.
} $P_{\alpha}(t)$ denote the configuration of $P_{\alpha}$ at time $t$. By
\cref{def:anoagg_processes}, $\bigl( P_{\alpha}(t) \bigr)_{t \geq 0}$ is a
Markov chain, since $P_{\alpha}(t)$ depends only on $P_{\alpha}(t-1)$. Another
immediate consequence of \cref{def:anoagg_processes} is that $ P_{\alpha}(t)$
conditional on $P_{\alpha}(t-1)=c$ is distributed according to $\Mult\bigl(n,
\alpha(c)\bigr)$. In other words, the 1-step distribution of an \AC/-process is
a multinomial distribution. Two important examples of \AC/-processes include
\voter and \hmajority{3}:
\begin{itemize}
\item
    In the \voter process $P_{\aggfvoter}$, each node samples one node
    (according to the pull mechanism) and (always) adopts that node's opinion.
    Thus
    \begin{equation}
    \label{eq:poll_up}
       \aggfvoter_i(c)
    =  \frac{c_i}{n}
    .
    \end{equation}

\item
    In the \hmajority{3} process $P_{\aggfhmajority{3}}$, each node samples
    independently and uniformly at random three nodes. If a color is supported
    by at least two of the samples, adopt it. Otherwise, adopt a random one of
    the sampled colors. Simple calculations (see~\cite{BCNPST14}) show
    \begin{equation}
      \aggfhmajority{3}_i(c)
    = \frac{c_i}{n} \cdot \left(1 + \frac{c_i}{n}-\norm{\frac cn}_2^2\right)
    .
    \label{eq:maj_up}
    \end{equation}
\end{itemize}

For any protocol $P$ starting with configuration $c\in \mathcal{C}$ let $T_P^\kappa(c)$ denote the first time step where the number of remaining
colors reduces to $\kappa$ where $\kappa \in \N$.
The next definition introduces dominance between protocol. Intuitively,
a protocol $P$ dominates another protocol $P'$ if their expected behavior
preserves majorization.
\begin{definition}[Protocol Dominance]
\label{def:prot_dom}
Consider two (not necessarily \AC/-) processes $P, \tilde{P}$. We say $P$
dominates $\tilde{P}$ if $\Ex{P(c)} \succeq \ex{\tilde{P}(\tilde{c})}$ holds for
all $c, \tilde{c} \in \mathcal{C}$ with $c \succeq \tilde{c}$.
\end{definition}
Note that, in the case of \AC/-protocols, \cref{def:prot_dom} can be stated
as follows: $P_{\alpha}$ dominates $P_{\tilde{\alpha}}$ if and only if $c
\succeq \tilde{c} \Rightarrow \alpha(c) \succeq \tilde{\alpha}(\tilde{c})$ for
all $c, \tilde{c} \in \mathcal{C}$ with $c \succeq \tilde{c}$. With this, the
main result of our framework can be stated as follows.
\begin{theorem}
\label{cor:anonymous_coupling}
Consider two \AC/-Processes $P$ and $P'$ where $P$ dominates $P'$. Assume
$P$ and $P'$ are started from the same configuration $c \in \mathcal{C}$. Then,
for any $\kappa \in \N$, the time needed by $P'$ to reduce the number of remaining
colors to $\kappa$ dominates the time $P$ needs for this, i.e.,
\[ T_{P'}^\kappa(c) \geq^{\text{st}} T_P^\kappa(c) .\]
\end{theorem}
One should note that the statement of \cref{cor:anonymous_coupling} is not
true in general (i.e., for non-\AC/-processes). In particular, \twochoice
dominates \voter, but our upper bound on \voter (\cref{sec:boundonvoter}) and
our lower bound on \twochoice (\cref{lem:lowerTCstrong}) contradict the
statement of \cref{cor:anonymous_coupling}.

\subsection{Coupling two \AC/-Processes}
\label{sec:couplingtwoACPs}

In order to prove \cref{cor:anonymous_coupling}, we formulate a strong 1-step
coupling property for \AC/-processes:
\begin{lemma}[1-Step Coupling]
\label{thm:anonymous_1step_coupling}
Let $P_{\alpha}$ and $P_{\tilde{\alpha}}$ be two \AC/-processes. Consider any
two configurations $c,\tilde{c} \in \mathcal{C}$ with $\alpha(c) \succeq
\tilde{\alpha}(\tilde{c})$.
 Let $c'$ and $\tilde{c}'$ be the configurations of $P_{\alpha}$
and $P_{\tilde{\alpha}}$ after one round, respectively.
Then, there exists a coupling such that $c' \succeq
\tilde{c}'$.
\end{lemma}
Note that \cref{cor:anonymous_coupling} is an immediate consequence of
\cref{thm:anonymous_1step_coupling}: Since $P$ dominates $P'$ (which is, for
\AC/-processes, equivalent to $\alpha(c) \succeq \tilde{\alpha}(\tilde{c})$ for
all $c, \tilde{c}$ with $c \succeq \tilde{c}$) we can apply
\cref{thm:anonymous_1step_coupling} iteratively to get
\cref{cor:anonymous_coupling}. The fine-grained comparison enabled by
\cref{thm:anonymous_1step_coupling} is based on three observations:
\begin{enumerate}
\item
    The (pre-) order \enquote{$\preceq$} on the set of configurations naturally
    measures the closeness to consensus. Indeed, a configuration with only one
    remaining color is maximal with respect to \enquote{$\preceq$}. Similarly,
    the $n$-color configuration is minimal.

\item
    We can define a vector variant \enquote{$\preceq^{\text{st}}$} of stochastic
    domination (see \cref{def:stochastic_majorization}) such that $\Theta_1
    \preceq \Theta_2 \Rightarrow \Mult\bigl( m, \Theta_1 \bigr)
    \preceq^{\text{st}} \Mult\bigl( m, \Theta_2 \bigr)$
    (\cite[Proposition~11.E.11]{Marshall:2011aa} or
    \cref{lem:multinomial:schur_convex} in \cref{sec:techtools}).

\item
    Consider two configurations $c,\tilde{c} \in \mathcal{C}$ with $\alpha(c)
    \succeq \tilde{\alpha}(\tilde{c})$. Since $P_{\alpha}(c) \sim \Mult(n,
    \alpha(c))$ and $P_{\tilde{\alpha}} \sim \Mult(n,
    \tilde{\alpha}(\tilde{c}))$, the previous observations imply that one step
    of $P_{\alpha}$ on $c$ is stochastically \enquote{better} than one step of
    $P_{\tilde{\alpha}}$ on $\tilde{c}$. Our goal is to apply
    \cref{thm:anonymous_1step_coupling} iteratively to get
    \cref{cor:anonymous_coupling}. For this, we prove a coupling showing
    majorization between the resulting configurations. We achieve this via a
    variant of Strassen's Theorem (see \cref{thm:strassen:simplified} below),
    which translates \emph{stochastic} domination among random vectors to the
    \emph{existence} of such a coupling.

\end{enumerate}

We now give a definition of stochastic majorization that is compatible with the
preorder \enquote{$\preceq$} on the configuration space $\mathcal{C}$
(cf.~\cite[Chapter~11]{Marshall:2011aa}).
\begin{definition}[Stochastic Majorization]
\label{def:stochastic_majorization}
For two random vectors $X$ and $Y$ in $\R^d$, we write $X \preceq^{\text{st}} Y$
and say that $Y$ \emph{stochastically majorizes} $X$ if $\Ex{\phi(X)} \leq
\Ex{\phi(Y)}$ for all Schur-convex functions $\phi$ on $\R^d$ such that the
expectations are defined.
\end{definition}
We proceed by stating the aforementioned variant (\cref{thm:strassen:simplified}) of Strassen's Theorem (\cref{thm:strassen:original}) whose derivation is provided in \cref{sec:Weihnachten}.
\begin{theorem}[Strassen's Theorem (variant)]
\label{thm:strassen:simplified}
Consider a closed subset $\mathcal{A} \subseteq \R^n$ such that the set $\set{
(x,y) | x \preceq y }$ is closed.
For two random vectors $X$ and $Y$ over $\mathcal{A}$, the following conditions
are equivalent:
\begin{enumerate}[(i)]
\item\label{thm:strassen:simplified:cond1}
    (Stochastic Majorization) $X \preceq^{\text{st}} Y$ and
\item\label{thm:strassen:simplified:cond2}
    (Coupling) there is a coupling between $X$ and $Y$ such that $\Pr{X \preceq
    Y} = 1$.
\end{enumerate}
\end{theorem}
With this, \cref{thm:anonymous_1step_coupling} follows by a straightforward
combination of the aforementioned machinery. See \cref{weltanschauung} for
details.


\newcommand{\xmin}{x_{\min}}
\newcommand{\xmax}{x_{\max}}\newcommand{\filt}{\mathcal{F}_t}
\newcommand{\filtr}{\mathcal{F}_r}

\section{Upper Bound for \hmajority{3}}
\label{sec:upper_bound_3majority}

In this section, we provide a sublinear upper bound on the time
needed by \hmajority{3} to reach consensus with high probability. This is one
of our main results and is formulated in the following theorem.
\begin{theorem}
\label{thm:3maj}
Starting from \emph{any} configuration $c \in \mathcal{C}$, \hmajority{3}
reaches consensus w.h.p. in $\LDAUOmicron[small]{ n^{3/4}\log^{7/8} n }$ rounds.
\end{theorem}
The analysis is split into two phases, each consisting of
$\LDAUOmicron[small]{n^{3/4}\log^{7/8} n}$ rounds.
\begin{description}
\item[Phase~1: From up to $n$ to $n^{1/4}\log^{1/8}$ colors.]
    This is the crucial part of the analysis. Instead of analyzing \hmajority{3}
    directly, we use our machinery from \cref{sec:anoagg} to show that
    \hmajority{3} is not slower than \voter (\cref{cor:3-majvspolling}).
    Then, we prove that \voter reaches $\LDAUOmicron[small]{n^{1/4}}$ colors in
    $\LDAUOmicron[small]{ n^{3/4}\log^{7/8} n }$ rounds (\cref{lem:reducetok}).

\item[Phase~2: From up to $n^{1/4}\log^{1/8}  n$ to $1$ color (consensus).]
    Once we reached a configuration with $n^{1/4}\log^{1/8}  n$ colors, we can
    apply~\cite[Theorem 3.1]{BCNPT16} (see
    \cref{thm:SanpellegrinoWhoGetsTheJoke} in \cref{sec:techtools}), a
    previous analysis of \hmajority{3}. It works only for initial configurations
    with at most $k \leq n^{1/3 - \epsilon}$ colors ($\epsilon > 0$ arbitrarily
    small). In that case, \cite[Theorem 3.1]{BCNPT16} yields a runtime of
    $\LDAUOmicron[small]{(k^2 \log^{1/2} n + k \log n) \cdot (k+ \log n)}$. Since the first phase leaves us with
    $\LDAUOmicron[small]{n^{1/4}}$ colors, this immediately implies that the second
    phase takes $\LDAUOmicron[small]{ n^{3/4}\log^{7/8} n }$ rounds.

\end{description}
This section proceeds by proving the runtime of Phase~1 in two steps: dominating
the runtime of \hmajority{3} by that of \voter (\cref{sec:3majvspolling}) and
proving the corresponding runtime for \voter (\cref{sec:boundonvoter}).
Afterwards, \cref{sec:3majbound:finishing} combines these results together
with~\cite[Theorem 3.1]{BCNPT16} to prove \cref{thm:3maj}.

\subsection{Analysis of Phase~1: \hmajority{3} vs.~\voter}
\label{sec:3majvspolling}

We prove the following lemma.
\begin{lemma}
\label{cor:3-majvspolling}
Consider \voter ($\mathcal{V}$) and \hmajority{3} ($3M$) started from the same initial configuration $c
\in \mathcal{C}$. There is a coupling such that after any round, the number of
remaining colors in \voter is not smaller than those in \hmajority{3}. In
particular, the time \voter needs to reach consensus stochastically dominates
the time needed by \hmajority{3} to reach consensus, i.e.,
\[ T_{3M}^\kappa(c) \leq^{\text{st}} T_{\mathcal{V}}^\kappa(c). \]
\end{lemma}
\begin{proof}
By \cref{cor:anonymous_coupling}, all we have to prove is $c \succeq
\tilde{c} \Rightarrow \aggfhmajority{3}(c) \succeq \aggfvoter(\tilde{c})$ (see \cref{sec:anoagg}). To
this end, consider two configurations $c, \tilde{c} \in \mathcal{C}$ with $c
\succeq \tilde{c}$. Let $p \coloneqq \aggfhmajority{3}(c)$ and $\tilde{p}
\coloneqq \aggfvoter(\tilde{c})$. We have to show $p \succeq \tilde{p}$. Since
these are probability vectors, we have $\norm{p}_1 = 1 = \norm{\tilde{p}}_1$.
It remains to consider the partial sums for $k \in \intcc{n}$. For this, let $x
\coloneqq c / n$ and $\tilde{x} \coloneqq \tilde{c} / n$. Remember that $p_i =
x_i^2 + (1-\norm{x}_2^2) \cdot x_i$ (\cref{eq:maj_up}) and $\tilde{p}_i =
\tilde{x}_i$ (\cref{eq:poll_up}). In the following, we assume (w.l.o.g.) $p
= \sortDown{p}$ and $\tilde{p} = \sortDown{\tilde{p}}$ (this implies $x =
\sortDown{x}$ and $\tilde{x} = \sortDown{\tilde{x}}$). We compute
\begin{equation}
\begin{aligned}
        \sum_{i=1}^{k} p_i
      - \sum_{i=1}^{k} \tilde{p}_i
&=      \sum_{i=1}^{k} x_i^2
      + \sum_{i=1}^{k} x_i
      - \norm{x}_2^2 \sum_{i=1}^{k} x_i
      - \sum_{i=1}^{k} \tilde{x}_i\\
&\geq   \sum_{i=1}^{k} x_i^2
      - \norm{x}_2^2 \sum_{i=1}^{k} x_i
.
\end{aligned}
\end{equation}
We have to show that this last expression is non-negative, which is equivalent
to
\begin{equation}
      \norm{x}_2^2
\leq  \left( \sum_{i=1}^{k} x_i^2 \right) / \left( \sum_{i=1}^{k} x_i \right)
.
\end{equation}
This holds trivially for $k = n$ (where we have equality). Thus, it is
sufficient to show that $( \sum_{i=1}^{k} x_i^2 ) / ( \sum_{i=1}^{k} x_i )$ is
non-increasing in $k$. That is, for any $k \in \intcc{n-1}$ we seek to show the
inequality
\begin{equation}
\label{eqn:3majvspol:toshow}
     \frac{ \sum_{i=1}^{k+1} x_i^2 }{ \sum_{i=1}^{k+1} x_i }
=    \frac{ \sum_{i=1}^{k} x_i^2 + x_{k+1}^2 }{ \sum_{i=1}^{k} x_i + x_{k+1} }
\leq \frac{ \sum_{i=1}^{k} x_i^2 }{ \sum_{i=1}^{k} x_i }
.
\end{equation}
This inequality is of the form $\frac{A + x}{B + x} \leq \frac{A}{B}$, where $A,
B, x > 0$. Rearranging shows that this is equivalent to $x \leq A / B$. Thus,
\cref{eqn:3majvspol:toshow} holds if and only if $x_{k+1} \leq (
\sum_{i=1}^{k} x_i^2 ) / ( \sum_{i=1}^{k} x_i )$. This last inequality holds via
\begin{math}
     x_{k+1} \cdot \sum_{i=1}^{k} x_{i}
=    \sum_{i=1}^{k} x_{i} \cdot x_{k+1}
\leq \sum_{i=1}^{k} x_{i} \cdot x_{i}
=    \sum_{i=1}^{k} x_{i}^{2}
\end{math},
where we used $x = \sortDown{x}$. This finishes the proof.
\end{proof}

\subsection{Analysis of Phase~1: A Bound for \voter}
\label{sec:boundonvoter}

We analyze the time the  \voter process takes to reduce the number of remaining colors
from $n$ to $k$. One should note that~\cite{BGKM16} studies  a similar process.
However, their analysis relies critically on the fact that their process is lazy
(i.e., nodes do not sample another node with probability $1/2$), while  our proof does
not require any laziness.

We make use of the well-known duality (via time reversal) between the
\voter process and \emph{coalescing random walks}.
In the coalescing random walks process there are initially $n$ independent random walks,
one placed at each of the $n$ nodes. While performing synchronous steps,
whenever two or more random walks meet, they coalesce into a single random walk.
Let $T_C^k$ denote the number of steps it takes to reduce the number of random
walks from $n$ to $k$ in the coalescing random walks process (the
\emph{coalescence time}). Similarly, let $T_{\mathcal V}^k$ denote the number of rounds it
takes \voter to reduce the number of remaining colors from $n$ to $k$.

Here we only sketch the main ideas of the proof, while the
full proof is given in \cref{topfpflanze}.

\begin{lemma}
\label{lem:reducetok}
Consider an arbitrary initial configuration $c \in \mathcal{C}$. \voter reaches
a configuration $c'$ having at most $k$ remaining colors w.h.p. in
$\LDAUOmicron[small]{\frac n k \, \log n}$ rounds, i.e.,
$\Pr{ T_\mathcal{V}^k =  \LDAUOmicron[small]{\frac n k \, \log n}} \geq 1-1/n $.
\end{lemma}
\begin{proof}[Sketch of Proof.]
%
    It is
    well-known (e.g., \cite{AF14}), that $T_{\mathcal V}^1=T_C^1$.
    This statement generalizes for all $k \in \intcc{n}$
    (see \cref{lem:coalvote} in \cref{sec:techtools} for a proof)
    to
    \begin{equation}
        T_{\mathcal V}^k=T_C^k.
        \label{Eichhoernchenkuchen}
    \end{equation}
    Thanks to the previous identity, we
    can prove the lemma's statement by proving that w.h.p. $T_C^k = \LDAUOmicron{ \frac n
     k \log n}$. To this end, we show that $\Ex{T_C^k} = \LDAUOmicron{n / k}$.
     From this we can easily derive a high probability statement (see the full proof in \cref{topfpflanze}).
     We now analyze the coalescing random walks.

    Let $X_t$ denote the number of coalescing random walks at time $t$. We have
    $X_0 = n$ and $T_C^k = \min\set{ t \geq 0 | X_t \leq k }$.
    One can argue (see the full proof for details) that in expectation
     \begin{align}\label{Hundefutter}
     \E{X_{t+1}~|~X_t=x} \leq x-  \frac{x^2}{10 n} .
     \end{align}
   Using this expected drop together with a
   drift theorem (\cref{lem:theoverlooked} in Appendix
   \ref{verdammten}) we finally get
   $ \E{T_C^k} \leq 20 \frac nk $.
   Using this and \cref{Eichhoernchenkuchen} we finally derive
   $
   \E{T_{\mathcal V}^k} =\E{T_C^k} \leq 20 \frac nk.
   $
   \qed


\section{Lower Bound for \twochoice}
\label{sec:lower_bound_2choice}
This section gives an almost linear worst-case lower bound on the time needed by
\twochoice to reach consensus with high probability. It turns out that, when
started from an almost balanced configuration, the consensus time is dictated by
the time it takes for one of the colors to gain a support of $\LDAUOmega{\log
n}$.
To prove this result, we prove a slightly stronger statement, that captures the
slow initial part of the process when started from configurations with a maximal
load of $\ell$. Here we only provide a sketch of proof. The full proof is given in \cref{sec:Dominik}
\begin{theorem}
\label{lem:lowerTCstrong}
Let $\gamma$ be a sufficiently large constant.
Consider the \twochoice process
starting from any initial configuration $c \in \mathcal{C}$. Let $\ell \coloneqq \max_i c_i(0)$ be the support of the largest color.
Then, for  $\ell' \coloneqq \max\{ 2\ell, \gamma \log n\}$, it holds with high probability that no color has a support  larger than $\ell'$ for $n/(\gamma \ell')$ rounds. In symbols,
\begin{equation}
    \Pr{ \max_i\hspace{2pt} c_i(t)   > \ell'\text{ for some }  t< \frac{n}{\gamma\ell'}} \leq \frac 1n.
\end{equation}
In particular, starting from the $n$-color configuration, 
it holds with high probability that no color has a support  larger than $\gamma \log n$ for $  \frac{n}{\gamma^2\log n}$ rounds.

\end{theorem}
\begin{proof}[Sketch of Proof]
Let $T_i=\min \{ t\geq 0 ~|~ c_i(t) > \ell' \}$.
For any fixed opinion $i\leq k$ we show that $\Pr{T_i < n/(\gamma \ell')} \leq
1/n^{2}$, so that, by a union bound over all opinions and using that $T=\min \{
T_i ~|~i\leq k\}$, we obtain $\Pr{T < n/(\gamma \ell')} \leq 1/n$.
Intuitively, we would like to show that, conditioning on $c_i  \leq \ell'$, the
expected number of nodes joining opinion $i$ is dominated by a binomial
distribution with parameters $n$ and $p = (\ell'/n)^2$.
The main obstacle to this is that naïvely applying Chernoff bounds for every time step
yields a weak bound, since with constant probability at each round at least one
color increases its support by a constant number of nodes.
Instead, we consider a new process $\mathcal P$ in which the number $P(t) $ of
nodes supporting color $i$ at time $t$ majorizes $c_i(t)$ as long as $P(t)
\leq \ell'$; we will then show that, after a certain time w.h.p.
$P(t)$ is still smaller than $\ell'$ implying that $\mathcal{P}$ indeed majorizes the original process. 
Using the fact that in $\mathcal{P}$ we can simply apply Chernoff
bounds over several rounds, we can finally get $c_i \leq P(t) \leq \ell'$ w.h.p.. 

Formally, process $\mathcal P$ is defined as follows. $P(0):=\ell$ and $P(t+1)
=P(t)+\sum_{j \leq n} X_j^{(t)}$, where $X_j^{(t)}$ is a Bernoulli random variable with
$\Pr{X_i=1} = p$ and, by a standard coupling, it 
is $1$ whenever node $j$ sees two times color $i$ at round $t$ (note that the latter event
happens with probability at most $p$ for any $t<T_i$).
By definition, if $t<T_i$ it holds $c_i(t)\leq \ell'$, which implies that the probability that any node in the original process gets opinion $i$ is at most $p$. Thus, we can couple
\twochoice and $P$ for $t\leq T_i$ so that $c_i(t) \leq P(t)$. This implies
that 
\begin{equation*}
    T' \coloneqq \min\{ t\geq 0 ~|~ P(t) \geq \ell'\}\preceq T_i.
\end{equation*} 
Relying on Chernoff bounds, we show in the full proof (\cref{sec:Dominik})  that $\Pr{T' < n/(\gamma \ell')}< 1/n^{2}$,
and from this we derive
$    \Pr{ T< n/(\gamma\ell') } 
\leq n^{-2}.
$
\end{proof}


\section{Conclusion \& Future Work}
\label{sec:hardtoget}

This section briefly discusses some directions of future work and our conjecture
that our framework might help to gain a better understanding of how different
(\AC/-) processes compare to each other.

\paragraph{Fault Tolerance}
As mentioned in the introduction, previous studies~\cite{BCNPST14, BCNPT16,
CER14, EFKMT16} show that \twochoice and \hmajority{3}  are 
consensus protocols     that can tolerate   dynamic,  worst-case adversarial faults.
More in details, the protocols work even in the presence  of an adversary that
 can, in every round, corrupt the state of a bounded set of
nodes. The goal in this setting is to achieve a stable regime in which
\enquote{almost-all} nodes support the same \emph{valid} color (i.e. a color initially
supported by at least one non-corrupted node). The size of the corrupted set is
one of the studied quality parameters and depends on the number $k$ of colors
and/or on the bias in the starting configuration. For instance,
in~\cite{BCNPT16} it is proven that, for $k = \LDAUomicron{n^{1/3}}$,
\hmajority{3} tolerates a corrupted sets of size $\LDAUOmicron{ \sqrt{n} / (
k^{5/2} \log n ) }$. A natural important open issue is to investigate whether
our framework for \AC/-processes can be used to make statements about
fault-tolerance properties in this (or in similar) adversarial models. We moderately lean toward
thinking that our analysis is sufficiently general and \enquote{robust} to be
suitably adapted in order to cope with this 
adversarial scenario over a wider range of $k$ and bias w.r.t. the relative previous analyses.

\paragraph{Towards a Hierarchy}
Consider the process functions of the general \hmajority{h} process for
arbitrary $h \in \N$. Intuitively, \hmajority{h} should be (stochastically)
slower than \hmajority{(h+1)}. We strongly believe this result holds. However,
naïvely applying our machinery to prove this does not work and needs to be
amended. Our conjecture that such a \enquote{hierarchy} for
$\hmajority{h}$ for different $h \in \N$ holds is backed by the proof of
\cref{cor:3-majvspolling} (which shows this for $h \in \set{1,2,3}$, since
the \voter process is actually equivalent to \hmajority{1} and \hmajority{2}).
\begin{conjecture}
\label{lem:hmajority_leq_h+1majority}
For $h \in \N$, we can couple \hmajority{h} and \hmajority{(h+1)} such that the
latter never has more remaining colors than the former. In particular,
\hmajority{(h+1)} is stochastically faster than \hmajority{h}.
\end{conjecture}

However, as we show in \cref{sec:counterexample} via a counterexample, it
turns out that \cref{thm:anonymous_1step_coupling} is not strong enough to
derive \cref{lem:hmajority_leq_h+1majority}. In fact, our failed attempts in
adapting our approach may suggest that similar counterexamples exist for any
majorization attempt that uses a total order on vectors.

\bibliographystyle{alpha}
\bibliography{references}

\appendix

\section{Auxiliary Tools and Full Proofs}
\label{sec:techtools}

\subsection{Tools from Majorization Theory}

\begin{proposition}[{\cite[Proposition~11.E.11]{Marshall:2011aa}},\cite{Rinott:1977aa}]
\label{lem:multinomial:schur_convex}
For $N \in \N$ and a probability vector $\Theta \in \intcc{0,1}^l$, consider a
random vector $X$ having the multinomial distribution $\Mult(N,\Theta)$. Let
\begin{equation}
\phi\colon \Set{ x \in \N_0^l | \sum_{ i \in \intcc{l} } x_i = N } \to \R
\label{eq:afunc}
\end{equation}
be such that $\Ex{\phi(X)}$ exists. Note that this expected value depends on
$\Theta$. Define the function $\psi$ on probability vectors as $\psi(\Theta)
\coloneqq \Ex{\phi(X)}$. If $\phi$ is Schur-convex, then so is $\psi$.
\end{proposition}

\begin{theorem}[{Strassen's Theorem~\cite[17.B.6]{Marshall:2011aa}}]
\label{thm:strassen:original}
Suppose that $\mathcal{A} \subseteq \R^n$ is closed and that
$\leq_{\mathcal{C}}$ is the preorder of $\mathcal{A}$ generated by the convex
cone $\mathcal{C}$ of real-valued functions defined on $\mathcal{A}$. Suppose
further that $\set{ (x,y) | x \leq_{\mathcal{C}} y }$ is a closed set. Then the
conditions
\begin{enumerate}[(i)]
\item\label{thm:strassen:original:cond1}
    $X \leq_{\mathcal{C}}^{\text{st}} Y$ and
\item\label{thm:strassen:original:cond2}
    there exists a pair $\tilde{X}, \tilde{Y}$ of random variables such that
    \begin{enumerate}[(a)]
    \item
        $X$ and $\tilde{X}$ are identically distributed, $Y$ and $\tilde{Y}$ are
        identically distributed and
    \item
        $\Pr{ \tilde{X} \leq_{\mathcal{C}} \tilde{Y} } = 1$
    \end{enumerate}
\end{enumerate}
are equivalent if and only if $\mathcal{C}^+ = \mathcal{C}^*$; i.e., the
stochastic completion $\mathcal{C}^+$ of $\mathcal{C}$ is complete.
\end{theorem}
\subsection{Proof of \cref{thm:strassen:simplified}}\label{sec:Weihnachten}
%
Consider the cone \[\mathcal{C} \coloneqq \set{ \phi\colon \mathcal{A} \to \R |
\phi \text{ is Schur-convex} }\] of real-valued Schur-convex functions on
$\mathcal{A}$. This cone implies a preorder \enquote{$\leq_{\mathcal{C}}$} on
$\mathcal{A}$ by the definition $x \leq_{\mathcal{C}} y
\mathrel{\vcentcolon\Leftrightarrow} \phi(x) \leq \phi(y)$ for all $\phi \in
\mathcal{C}$. One can show that this preorder is the vector majorization
\enquote{$\preceq$} (cf.~\cite[Example~14.E.5]{Marshall:2011aa})\footnote{%
    Alternatively, one checks this manually: The direction $x \preceq y
    \Rightarrow x \leq_{\mathcal{C}} y$ is trivial by the definition of
    Schur-convexity. For $x \leq_{\mathcal{C}} y \Rightarrow x \preceq y$
    consider the $n+1$ Schur-convex functions $z \mapsto \sum_{j \in
    \intcc{i}}\sortDown{z}$ for $i \in \intcc{n}$ and $z \mapsto -\norm{z}_1$.
}.
Now, \enquote{$\leq_{\mathcal{C}}$} being equal to \enquote{$\preceq$} has two
implications:
\begin{enumerate}[(a)]
\item\label{thm:strassen:simplified:imp1}
    The stochastic majorization \enquote{$\leq_{\mathcal{C}}^{\text{st}}$}
    implied by the preorder \enquote{$\leq_{\mathcal{C}}$} is the stochastic
    majorization \enquote{$\leq^{\text{st}}$} from
    \cref{def:stochastic_majorization}
    (cf.~\cite[Definition~17.B.1]{Marshall:2011aa}).
\item\label{thm:strassen:simplified:imp2}
    Since a cone $\mathcal{C}$ is complete if it is maximal with respect to
    functions preserving the preorder \enquote{$\leq_{\mathcal{C}}$}
    (cf.~\cite[Definition~14.E.2]{Marshall:2011aa}), $\mathcal{C}$ is complete
    (Schur-convex functions are by definition the set of all functions
    preserving the majorization preorder).
\end{enumerate}
From \ref{thm:strassen:simplified:imp1} we get that
Condition~\ref{thm:strassen:simplified:cond1} is actually
Condition~\ref{thm:strassen:original:cond1} of \cref{thm:strassen:original}.
The same holds for Condition~\ref{thm:strassen:simplified:cond2}. From
\ref{thm:strassen:simplified:imp2} we get that
$\mathcal{C}=\mathcal{C}^*=\mathcal{C}^+$
(cf.~\cite[Proposition~17.B.3]{Marshall:2011aa}), such that
Conditions~\ref{thm:strassen:simplified:cond1}
and~\ref{thm:strassen:simplified:cond2} are equivalent by
\cref{thm:strassen:original}. This finishes the proof. \qed
%

\subsection{Proof of \cref{thm:anonymous_1step_coupling}}\label{weltanschauung}

Consider the processes $P_{\alpha}$ and $P_{\tilde{\alpha}}$ with the
configurations $c$ and $\tilde{c}$ from the theorem's statement. Let $Y=c'$ and
$X=\tilde{c}'$ denote the configurations resulting after one round of
$P_{\alpha}$ on $c$ and $P_{\tilde{\alpha}}$ on $\tilde{c}$, respectively. Let
$\Theta \coloneqq \alpha(c)$ and $\tilde{\Theta} \coloneqq
\tilde{\alpha}(\tilde{c})$. As observed earlier in \cref{sec:anoagg}, we have
$Y \sim \Mult(n,\Theta_1)$ and $X \sim \Mult(n,\Theta_2)$. By the theorem's
assumption, we have $\Theta \succeq \tilde{\Theta}$. Since, by
\cref{lem:multinomial:schur_convex} (see \cref{sec:techtools}), the
function $\Theta \to \Ex{\phi\bigl(\Mult(n,\Theta)\bigr)}$ is Schur-convex for
any Schur-convex function $\phi$ for which the expectation exists, we get $X
\preceq^{\text{st}} Y$.

Since the configuration space $\mathcal{C}$ is a finite subset of $\R^n$, it is
closed and so is $\set{ (x,y) | x \preceq y }$. We now apply
\cref{thm:strassen:simplified} (Strassen's Theorem, see
\cref{sec:techtools}) to get that there exists a coupling between $X$ and $Y$
such that\footnote{Observe that Strassen's Theorem gives us that $\Pr{ X \preceq Y } = 1$. That is, $X \preceq Y$ holds \emph{almost
surely}. However, since $\mathcal{C}$ is finite, this actually means that $\tilde{c}' = X
\preceq Y = c'$ holds (surely).} $X\preceq Y$. This finishes the proof.
\qed

\subsection{Tools from Drift Theory}\label{verdammten}

\begin{theorem}[Variable Drift Theorem {\cite[Corollary 1.(i)]{LW14}}]
\label{lem:theoverlooked}
Let $(X_t)_{t\ge 0}$, be a stochastic process over some state space $S\subseteq
\{0\}\cup [\xmin,\xmax]$, where $\xmin\ge 0$. Let $h\colon [\xmin,\xmax]\to\R^+$
be a differentiable function. Then the following statements hold for the first
hitting time $T \coloneqq \min\{t\mid X_t=0\}$.
If $\E{X_{t+1}-X_{t} \mid \filt; X_t\ge \xmin} \le -h(X_t)$ and $\frac{d }{d x} h(x)\geq 0$,
then
\begin{align*}
     \E{T\mid X_0}
\leq \frac{\xmin}{h(\xmin)} + \int_{\xmin}^{X_0} \frac{1}{h(y)} \,\mathrm{d}y
.
\end{align*}
\end{theorem}

\subsection{Tools for Consensus Processes}

\begin{theorem}[{\cite[Theorem 3.1]{BCNPT16}}]
\label{thm:SanpellegrinoWhoGetsTheJoke}
Let $\varepsilon >0$ be an arbitrarily small constant. Starting from any initial
configuration with $k \leq n^{1/3-\varepsilon}$ colors, \hmajority{3} reaches
consensus w.h.p. in
\[
    \LDAUOmicron{(k^2 \log^{1/2} n + k \log n) \cdot (k+ \log n)}
\]
rounds.
\end{theorem}

\begin{figure}
\centering
	\includegraphics[scale=0.7]{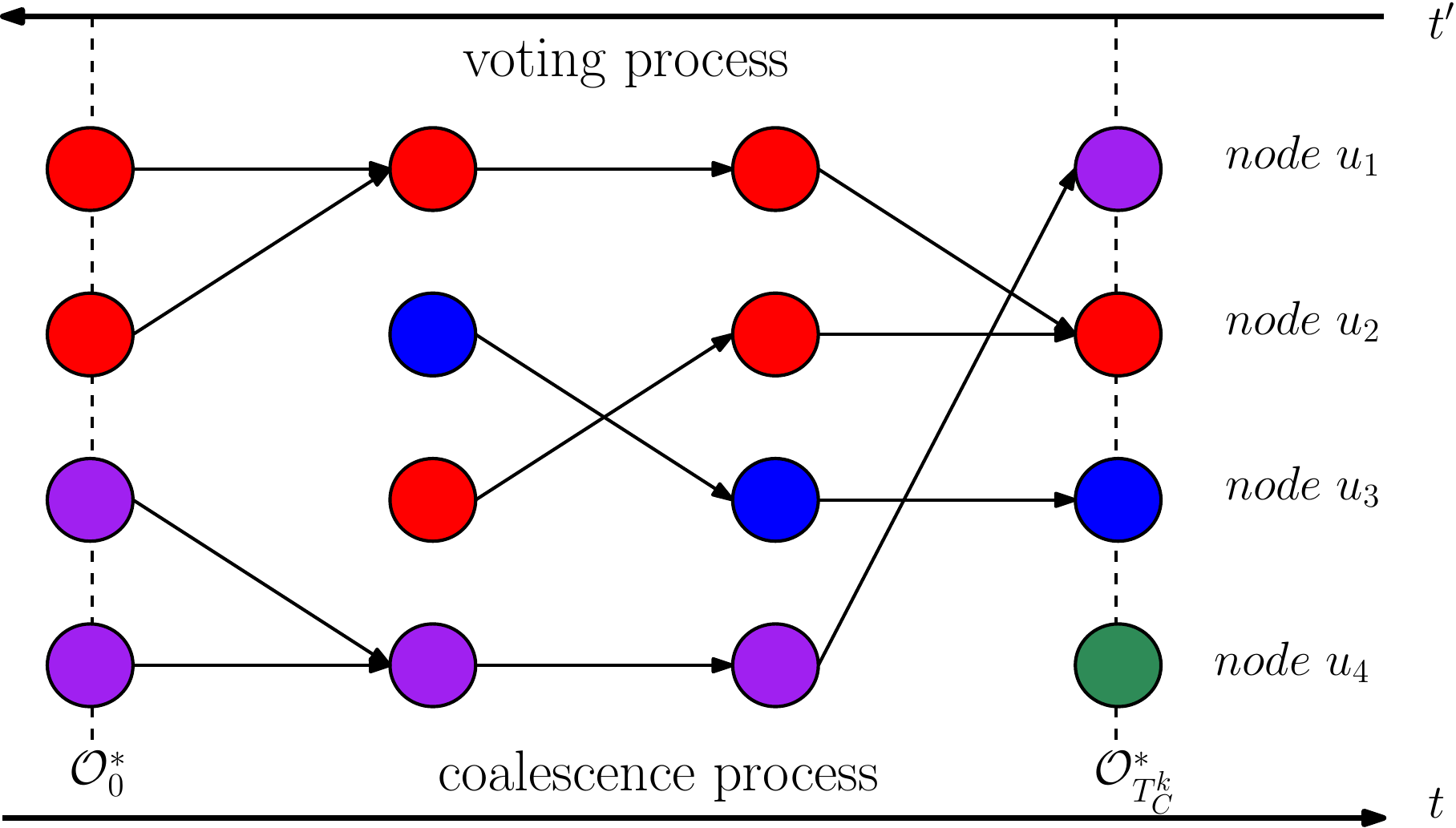}
\caption{Running the coalescence process from right to left (an edge from $u$ to $v$ means that the token on $u$ -if any- moves to $v$) yields that after $T=4$ rounds the number of random walks reduces to $k=2$. Using the same random choices (black arrows) for the voter process and running the process from left to right (an edge from $u$ to $v$ means that $u$ pulls $v$'s opinion) we derive that the number of opinions after $T=4$ rounds is also $2$. This is no coincidence as we show in \cref{lem:coalvote}. }
\label{itout}
\end{figure}
The following lemma uses the high-level idea of the proof presented in \cite[Chapter
14]{AF14} which only considers the case $k=1$.
For the purposes of our proof we would only require a coupling with $T_\mathcal{V}^k \leq T_C^k$, but for
the sake of completeness we show the stronger claim $T_\mathcal{V}^k = T_C^k$.
\begin{lemma}
\label{lem:coalvote}
For any graph $G = (V,E)$, there exists a coupling such that $T_C^k = T_{\mathcal V}^k $.
\end{lemma}
\begin{proof}
For $t\in \mathbb{N}$ and for $u\in V$
define the random variables $Y_{t}(u)$ with $Y_{t}(u) \sim
\operatorname{uniform}(N(u))$, where $\operatorname{uniform}(\cdot)$ denotes the
uniform distribution and $N(u)$ denotes the neighborhood of $u$. Hence,
$Y_{t}(u)=v$ means that $u$ pulls information from node $v$ in step $t$.
In the \textsc{Coalescence} process, the random variable $Y_{t}(u) \in N(u)$,  $t \in [0,
T_C^k)$ captures the transition performed by the random walk which is at $u$ at time
$t$ (if any). In other words, these random variables define the arrows in \cref{itout}.
For the voter process  $Y_{t}(u)=v$ means that in step $t$ node $u$ adopts the opinion of $v$.

Let $X(u)=(X_{0}(u)=u,X_1(u), \dots, X_{T_C^k}(u))$ be the trajectory of
the random walk starting at $u$.
We can thus express
\begin{equation}\label{Honigkuchen}
X_{t}(u) =
\begin{cases}
u & \text{ if $t=0$}\\
Y_{t-1}(X_{t-1}(u)) & \text{ otherwise.}
\end{cases}
\end{equation}
Thus, this trajectory $X(u)$ and the random variable $T_C^k$ are completely determined by the
random variables $\mathcal{Y}=\{ Y_{t}(u) \colon t \in \mathbb{N}, u\in V\}$.

%
%
%
Let $\mathcal{V}_{T_C^k}$ be the \voter process whose starting time $t'=0$ equals the time $T_C^k$ of the coalescence process (see also \cref{itout}).
Let $\mathcal{O}^{*}_{T_C^k - t'}(u)$ be the opinion of $u$ at time $t'$ of
$\mathcal{V}_{T_C^k}$.
For every node $u\in V$ and $t'\in [0,
T_C^k]$ we can thus express
\begin{equation}\label{Oblaten}
\mathcal{O}^{*}_{T_C^k-t'}(u) =
\begin{cases}
u & \text{ if $t'=0$}\\
\mathcal{O}^{*}_{ T_C^k-(t'-1)}(Y_{T_C^k-t'}(u)) & \text{ otherwise.}
\end{cases}
\end{equation}
Note that \eqref{Oblaten} constructs a coupling between the \voter process and the coalescence process through the common usage of the random variables $\mathcal{Y}$ in \eqref{Honigkuchen} and \eqref{Oblaten}.
 In particular, by unrolling \eqref{Honigkuchen} and \eqref{Oblaten} we get
\begin{align*}
X_{T_C^k}(u)&= Y_{T_C^k-1}(Y_{T_C^k-2}(\dots (Y_0(X_0(u)))\dots ))
\stackrel{(a)}{=} Y_{T_C^k-1}(Y_{T_C^k-2}(\dots (Y_0(u))\dots ))\\
O^*_0 (u)&= O^*_{T_C^k}(Y_{T_C^k-1}(Y_{T_C^k-2}(\dots (Y_0(u))\dots )))
\stackrel{(b)}{=}  Y_{T_C^k-1}(Y_{T_C^k-2}(\dots (Y_0(u)\dots )),
\end{align*}
where and $(a)$ we used that $X_0(u)=u$ and in $(b)$ we used that $O^{*}_{T_C^k}(v)=v$ for all $v$.
The above equations imply
\begin{equation}\label{Blau}
X_{T_C^k}(u)=\mathcal{O}^*_{0}(u).
\end{equation}

Let $Z_{t}=\{ X_{t}(u) \colon u\in V\}$
denote the positions of the remaining walks in the coalescence process at time $t$.
Observe that $|Z_0|=n$,
$|Z_{T_C^k}|\leq k$, by definition of $T_C^k$. We have, by \eqref{Blau}, that
\begin{align}\label{weinkeller}
  Z_{T_C^k}
= \{ X_{T_C^k}(u) \colon u\in V\}
= \{ \mathcal{O}^*_{0}(u)\colon u\in V\}
=: \mathcal{O}^*_{0} .
\end{align}
From \eqref{weinkeller}  we infer $|\mathcal{O}^{*}_{0}| = |Z_{T_C^k}|\leq k$, which  implies that
\[T_{\mathcal{V}}^k \leq T_{C}^k.\]
In the reminder we generalize the previous coupling to show that
\[T_{\mathcal{V}}^k = T_{C}^k.\]

%

In particular, we consider the \voter process for all starting position $\tau < T_C^k$ (all nodes have different colors at round $t$)  and show
that the resulting number of opinions is strictly more than $k$.

Let $\mathcal{V}_\tau$ be the \voter process that starts at time $\tau\in [0,T_C^k)$, and
let $\mathcal{O}^\tau_{T_C^k -t'}(u)$ be the opinion of $u$ at time $t'$ of
$\mathcal{V}_\tau$.  For every node $u\in V$ and $t'\in [0,
\tau]$ we have
\begin{equation}\label{Schnitzel}
\mathcal{O}^\tau_{\tau-t'}(u) =
\begin{cases}
u & \text{ if $t'=\tau$}\\
\mathcal{O}^\tau_{\tau-(t'-1)}(Y_{\tau-t'}(u)) & \text{ otherwise.}
\end{cases}
\end{equation}
Similarly as before,
  by unrolling \eqref{Honigkuchen} and \eqref{Schnitzel} we get
\begin{align*}
X_{\tau}(u)&= Y_{\tau-1}(Y_{\tau-2}(\dots (Y_0(X_0(u)))\dots ))
\stackrel{(a)}{=} Y_{\tau-1}(Y_{\tau-2}(\dots (Y_0(u))\dots ))\\
O^\tau_0 (u)&= O^\tau_{\tau}(Y_{\tau-1}(Y_{\tau-2}(\dots (Y_0(u))\dots )))
\stackrel{(b)}{=}  Y_{\tau-1}(Y_{\tau-2}(\dots (Y_0(u)\dots )),
\end{align*}
where and $(a)$ we used that $X_0(u)=u$ and in $(b)$ we used that $O^{\tau}_{\tau}(v)=v$ for all $v$.
%
%
%
 By defining  $\mathcal{O}^0_{T_C^k-t'}=\{
\mathcal{O}^0_{T_C^k-t'}(u) \colon u\in V\},$
from the above equations we get that
\[ X_{\tau}(u)=\mathcal{O}^\tau_{0}(u).\]
 Hence,
\begin{align}\label{wein}
  Z_\tau
= \{ X_{\tau}(u) \colon u\in V\}
= \{ \mathcal{O}^\tau_{0}(u)\colon u\in V\}
=: \mathcal{O}^\tau_{0} .
\end{align}
Since for all $\tau< T_C^k$ we have
$|Z_\tau|> k$,
from \eqref{wein} it follows that $|\mathcal{O}^{\tau}_{0}| =  |Z_{\tau}| > k$
 which yields the claim.
\end{proof}

\subsection{Proof of \cref{lem:reducetok}}
\label{topfpflanze}
We prove the lemma using the well-known duality (via time reversal) between the
\voter process and \emph{coalescing random walks}.
%

%
 It is
well-known (e.g., \cite{AF14}), that $T_{\mathcal V}^1=T_C^1$.
This statement generalizes for all $k \in \intcc{n}$
(see \cref{lem:coalvote} in \cref{sec:techtools} for a proof)
to
\begin{equation}
    T_{\mathcal V}^k=T_C^k.
    \label{Eichhoernchen}
\end{equation}
Thanks to the previous identity, we
can prove the lemma's statement by proving that w.h.p. $T_C^k = \LDAUOmicron{ \frac n
 k \log n}$. To this end, we show that $\Ex{T_C^k} = \LDAUOmicron{n / k}$. In order to get the claimed bound
in concentration, we can apply the following  standard argument.
Consider the process as a sequence of phases, each one  of
length $2\Ex{T_C^k}$. We say that a phase  is  successful when the number of remaining random
walks drops below $n / k$. Thanks to
our bound in expectation above and the   Markov inequality, we easily get that
every phase has probability $\Omega(1)$
to be successful.   So, with high probability, there will be at least one  success within the
first $\LDAUOmicron{\log n}$ phases.

Let $X_t$ denote the number of coalescing random walks at time $t$. We have
$X_0 = n$ and $T_C^k = \min\set{ t \geq 0 | X_t \leq k }$. We seek to apply
drift theory (\cref{lem:theoverlooked} in Appendix
\ref{verdammten}) to derive a bound on $\E{T_C^k}$.
Next, we compute an upper bound on $\E{X_{t+1} - X_t~|~X_t=x}$.

Let us begin assuming that $k$ is any constant.
It holds in general that
$\E{X_{t+1}-X_t \mid X_t\ge 2} \leq -1/n $, since in expectation two
random walks collide w.p.~$1/n$ in a given time step.
Hence we can directly apply\footnote{Technically, one would have to define a
new random variable which is $0$ whenever the number of random walks reduces to
$1$. We illustrate this technicality shortly, for case $k >100$ below.}
\cref{lem:theoverlooked} with parameters $h(x)=1/n$ to reduce from $k$
random walks to $1$, yielding the bound $\E{T_C^k} = O(n/k) = O(n)$, where in
the latter equality we used that $k$ is constant.

We now consider the case where $k$ is larger than a big constant, say $k > 100$.
Assume that in every time step the random walks move in two phases.
Let $W_1$ denote an arbitrary set of $\floor{X_t/2}$ random walks and let $W_2$
denote the remaining ones.
We first look at how the random walks in $W_1$ coalesce, then we consider
the movement of the remaining walks $W_2$.
Let $\mathcal{E}$ be the event that the walks in $W_1$ move onto more than
$\floor{X_t/4}$ distinct nodes.
This would imply that each walk in $W_2$ coalesces with one in $W_1$ with
probability at least ${\floor{X_t/4}}/{n}$.
We thus have
\[
    \E{X_{t+1}~|~X_t=x, \mathcal{E}}
    \leq x- \ceil{x/2} \cdot \frac{\floor{x/4}}{ n} \leq x-  \frac{x^2}{10 n}.
\]
Moreover, conditioning on $\overline{\mathcal{E}}$ implies that there were at
least $\floor{X_t/2}-\floor{X_t/4}$ collisions during the first phase. Thus,
\[
    \E{X_{t+1}~|~X_t=x, \overline{\mathcal{E}}} \leq x- (\floor{x/2}-\floor{x/4})\leq x-\frac{x^2}{10 n}.
\]
Hence, by law of total expectation,
\begin{align*}
    \E{X_{t+1}~|~X_t=x} &= \E{X_{t+1}~|~X_t=x,\mathcal{E}}\Pr{\mathcal{E}} +
        \E{X_{t+1}~|~X_t=x,\overline{ \mathcal{E}}}\Pr{\overline{\mathcal{E}}} \\
    &\leq x-  \frac{x^2}{10 n} .
\end{align*}
In order to apply \cref{lem:theoverlooked}, we define the random variables
$(Y_t)_{t\geq 0}$ as follows
\[
Y_t =\begin{cases}
X_t & \text{if $X_t > k$,}\\
0 & \text{otherwise.}
\end{cases}
\]

Let $T^*= \{ t\geq 0~|~ Y_t= 0\}$. Since by construction we have $Y_t= X_t$
for $t < T_C^k$ and $Y_{T_C^k}=0$ otherwise, it follows that
\begin{equation}
    T_C^k=T^*.
    \label{fremdschaemen}
\end{equation}
Therefore,
\[
    \E{Y_{t+1}~|~Y_t=y, Y_t >  k} \leq y - \frac{y^2}{10n}.
\]

We can thus apply \cref{lem:theoverlooked} for the random variables $(Y_t)_{t\geq
0}$ with $\xmin = k, \xmax =n$, and $h(x)=\frac{x^2}{10n}$, obtaining
\begin{align}
\E{T^*} &\leq \frac{k}{k^2/(10n)} + \int\limits_{k}^{n} \frac{1}{h(u)}
    \leq \frac{10n}{k} + 10n\left(-\frac{1}{n}-\left(-\frac{1}{k}\right)\right) \leq 20\frac nk.
    \label{fred_axt}
\end{align}
Finally, from \eqref{Eichhoernchen}, \eqref{fremdschaemen} and \eqref{fred_axt} we get
\begin{equation}
    \E{T_{\mathcal V}^k} = \E{T_C^k} = \E{T^*} \leq 20 \frac nk,
    \label{}
\end{equation}
concluding the proof.
\qed

\subsection{Proof of \cref{thm:3maj}}
\label{sec:3majbound:finishing}

Consider any initial configuration $c \in \mathcal{C}$. By applying
\cref{lem:reducetok} for $k = n^{1/4}$, we get that \voter reduces the number
of remaining colors w.h.p. from initially at most $n$ to $n^{1/4}$ in
$\LDAUOmicron[small]{ n^{3/4}\log^{7/8} n }$ rounds. By \cref{cor:3-majvspolling}, the time it
takes \hmajority{3} to reach some fixed number of remaining colors is dominated
by the time it takes \voter to reach the same number of remaining colors. In
particular, we get that \hmajority{3} also reduces the number of remaining
colors w.h.p. to $n^{1/4}$ in
$\LDAUOmicron[small]{ n^{3/4}\log^{7/8} n }$ rounds. That is, the first phase takes
$\LDAUOmicron[small]{ n^{3/4}\log^{7/8} n }$ rounds.

For the second phase, we apply~\cite[Theorem 3.1]{BCNPT16} (see
\cref{thm:SanpellegrinoWhoGetsTheJoke} in \cref{sec:techtools}) for $k =
n^{1/4} = \LDAUomicron{n^{1/3}}$. This immediately yields that the second phase
takes $\LDAUOmicron[small]{ n^{3/4}\log^{7/8} n }$ rounds, finishing the proof.
\qed

\subsection{Proof of \cref{lem:lowerTCstrong}}\label{sec:Dominik}
Let $T_i=\min \{ t\geq 0 ~|~ c_i(t) > \ell' \}$.
For any fixed opinion $i\leq k$ we show that $\Pr{T_i < n/(\gamma \ell')} \leq
1/n^{2}$, so that, by a union bound over all opinions and using that $T=\min \{
T_i ~|~i\leq k\}$, we obtain $\Pr{T < n/(\gamma \ell')} \leq 1/n$.
Intuitively, we would like to show that, conditioning on $c_i  \leq \ell'$, the
expected number of nodes joining opinion $i$ is dominated by a binomial
distribution with parameters $n$ and $p = (\ell'/n)^2$.
The main obstacle to this is that naïvely applying Chernoff bounds for every time step
yields a weak bound, since with constant probability at each round at least one
color increases its support by a constant number of nodes.
Instead, we consider a new process $\mathcal P$ in which the number $P(t) $ of
nodes supporting color $i$ at time $t$ majorizes $c_i(t)$ as long as $P(t)
\leq \ell'$; we will then show that, after a certain time w.h.p.
$P(t)$ is still smaller than $\ell'$ implying that $\mathcal{P}$ indeed majorizes the original process.
Using the fact that in $\mathcal{P}$ we can simply apply Chernoff
bounds over several rounds, we can finally get $c_i \leq P(t) \leq \ell'$ w.h.p..

Formally, process $\mathcal P$ is defined as follows. $P(0):=\ell$ and $P(t+1)
=P(t)+\sum_{j \leq n} X_j^{(t)}$, where $X_j^{(t)}$ is a Bernoulli random variable with
$\Pr{X_i=1} = p$ and, by a standard coupling, it
is $1$ whenever node $j$ sees two times color $i$ at round $t$ (note that the latter event
happens with probability at most $p$ for any $t<T_i$).
By definition, if $t<T_i$ it holds $c_i(t)\leq \ell'$, which implies that the probability that any node in the original process gets opinion $i$ is at most $p$. Thus, we can couple
\twochoice and $P$ for $t\leq T_i$ so that $c_i(t) \leq P(t)$. This implies
that
\begin{equation}
    T' \coloneqq \min\{ t\geq 0 ~|~ P(t) \geq \ell'\}\preceq T_i.
    \label{eq:majT}
\end{equation}

In the remainder we show that $\Pr{T' < n/(\gamma \ell')}< 1/n^{2}$.
For any round $t+1$, we define $\Delta_{t+1} \coloneqq P(t+1)-P(t) = \sum_{i
\leq n} X_i$.
Observe that
$\Delta_{t+1}\sim \operatorname{Bin}(n, p)$. 
%
%
Let $t_0 = n/(\gamma \ell')$. In the following we bound
\[
    B\coloneqq P( t_0
)-P(0)=\sum_{i=1}^{t_0} \Delta_i.\]
Observe that $B\sim \operatorname{Bin}(t_0 \cdot n, p)$ and thus
$ \E{B} = t_0 \cdot n \cdot p $.

Using Chernoff bounds,
e.g., \cite[Theorem 4.4]{MU05} we derive for any $\gamma \geq 18$
\begin{align}
    \Pr{ P(t_0) \geq \ell' }  &=\Pr{  B \geq \ell' -\ell }   \leq \Pr{B \geq \max\left\{2\E{B} , \frac{\gamma}{2}\log n \right\}} \nonumber\\
    &=\Pr{B \geq \E{B}\cdot \max\left\{2, 1+\frac{\frac{\gamma}{2}\log n}{\E{B}} \right\} } \leq \exp\left( -\frac{ \frac{\gamma}{2} \log n   }{3} \right) \leq 1/n^{3},
    \label{eq:boundP}
\end{align}
where we used that
\begin{align*}
    \max\left\{2\E{B} , \frac{\gamma}{2}\log n \right\}
    &=\max\left\{2t_0 \cdot n \cdot p , \frac{\gamma}{2}\log n \right\}\\
    &\leq \max \left\{ \frac{(\ell')^2}{\gamma \ell'} ,  \frac{\gamma}{2} \log n \right\} \\
    &\leq \max\left\{{\frac{\ell'}{2}}, \frac{\gamma}{2} \log n\right\}\\
        &\leq \frac{\ell'}{2}= \ell'-\ell.
    \label{eq:calc}
\end{align*}
Putting everything together yields
\begin{align}
    \Pr{ T< n/(\gamma\ell') }
    & = \Pr{ T< t_0 }\\
    &\stackrel{(a)}{\leq} n \Pr{T_i< t_0 }\nonumber\\
    &\stackrel{(b)}{\leq} n \Pr{T'< t_0 } \nonumber\\
    &\stackrel{(c)}{\leq} n \Pr{ P(t_0) \geq  \ell'} \stackrel{(d)}{\leq} n^{-2},
    \label{eq:qedlower}
\end{align}

where in $(a)$ we used union bound over all colors, in $(b)$ we used \eqref{eq:majT}, in $(c)$ we used that ``$T'< t_0$''$\implies$ ``$P(t_0) \geq  \ell'$'' and in $(d)$ we used \eqref{eq:boundP}.
This completes the proof.
\qed

\section{Limitations of \cref{thm:anonymous_1step_coupling} }\label{sec:counterexample}
In this section we show that there are
configurations $c \preceq \tilde{c}$ such that $\aggfhmajority{h}(c) \not\preceq
\aggfhmajority{h+1}(\tilde{c})$.
This means that, \cref{thm:anonymous_1step_coupling} is not strong enough
to derive \cref{lem:hmajority_leq_h+1majority}.
Consider the configurations $x\coloneqq (1/2, 1/6, 1/6, 1/6) \preceq (1/2, 1/2, 0, 0) \eqqcolon \tilde{x}$
(for simplicity, we use the fraction vectors $x = c/n$). For symmetry reasons,
we immediately get that $\aggfhmajority{h+1}(\tilde{c}) = (1/2, 1/2, 0, 0) =
\tilde{c}$. However, even for $h = 3$, for the second configuration we get that
the expected fraction of the nodes which adopt the first opinion after one step
is
\begin{equation}
    1           \cdot \binom{3}{0} \cdot {\left( \frac{1}{2} \right)}^3
  + 1           \cdot \binom{3}{1} \cdot {\left( \frac{1}{2} \right)}^2 \cdot \frac{3}{6}
  + \frac{1}{3} \cdot \binom{3}{2} \cdot \frac{1}{2} \cdot \frac{3}{6} \cdot \frac{2}{6}
= \frac{7}{12}
.
\end{equation}
The three terms of the sum on the left hand side correspond to the cases and
probabilities for which the first color is adopted:
\begin{itemize}
\item
    all samples choose color $1$ (probability to win is $1$, number of cases
    $\binom{3}{0}$),
\item
    two samples choose color $1$ (probability to win is $1$, number of cases
    $\binom{3}{1}$), or
\item
    1 sample chooses color $1$ and the other samples choose different colors
    (probability to win is $1/3$, number of cases $\binom{3}{2}$).
\end{itemize}
Thus, for $n$ large enough, with high probability the configuration resulting
from  \hmajority{(h+1)} will not majorize the one resulting from $\hmajority{h}$.

\end{document}